\documentclass[12pt]{article}
\usepackage{graphicx,subfigure}
\usepackage{graphicx}
\usepackage{amsfonts,amsmath}
\usepackage[mathscr]{eucal}
\usepackage{amssymb}
\usepackage{amsthm}
\usepackage{bbold}
\theoremstyle{plain}

\newtheorem{prop}{Proposition}

\newcommand{\be}{\begin{eqnarray}}
\newcommand{\ee}{\end{eqnarray}}
\newcommand{\bc}{\begin{center}}
\newcommand{\ec}{\end{center}}
\newcommand{\nn}{\nonumber \\}
\newcommand{\lb}{\label}
\newcommand{\p}[1]{(\ref{#1})}

\begin{document}

\begin{titlepage}

\vspace*{0.2cm}

\renewcommand{\thefootnote}{\star}
\begin{center}

{\LARGE\bf  On exactly solvable ghost-ridden systems}

\vspace{2cm}

{\Large A.V. Smilga} \\

\vspace{0.5cm}

{\it SUBATECH, Universit\'e de
Nantes,  4 rue Alfred Kastler, BP 20722, Nantes  44307, France. }

\end{center}
\vspace{0.2cm} \vskip 0.6truecm \nopagebreak

   \begin{abstract}
\noindent We discuss exactly solvable systems involving integrals of motion with higher powers of momenta. If one of these integrals is chosen for the Hamiltonian, we obtain a system involving {\it ghosts}, i.e. a system whose  Hamiltonian is not bounded neither from below, nor from above. However, these ghosts are {\it benign}: there is no collapse and unitarity is not violated.

 As an example, we consider the 3-particle Toda periodic chain, with the cubic invariant $I$ chosen for the Hamiltonian. The classical trajectories exhibit regular oscillations, and the spectrum of the quantum Hamiltonian is discrete running from $-\infty$ to $\infty$. 
 We also discuss the classical dynamics of a {\it perturbed} system with the Hamiltonian $H = I + v$, where $v$ is an oscillator potential. Such a system is not exactly solvable, but its classical trajectories exhibit not regular, but still benign behaviour without collapse. This means that also the corresponding quantum problem is well defined.  
 
 The same observation can be made for exactly solvable (1+1)-dimensional field theories involving an infinite number of conservation laws: any of them can be chosen for the Hamiltonian. We illustrate this for the Sine-Gordon and KdV models. In the latter case, the Lagrangian and standard integrals of motion involve higher spatial rather than temporal derivatives. But one can always interchange $x$ and $t$, after which we obtain a system with benign ghosts. 
 
   \end{abstract}

\end{titlepage}

\setcounter{footnote}{0}

\setcounter{equation}0

\section{Introduction}
The first example of a nontrivial interacting system which involves ghosts, but where ghosts are benign  was found in \cite{Robert}. Its Hamiltonian has the form
 \be
 \lb{HamRob}
 H = pP + DV'(x),
 \ee
 where $(p,x)$ and $(P,D)$ are two pairs of dynamical variables and $V(x)$ is an even polynomial. In the simplest case,
  \be
  \lb{V(x)}
  V(x) \ =\ \frac {\omega^2 x^2}2 + \frac {\lambda x^4}4.
   \ee
   Obviously, the sign of the classical Hamiltonian \p{HamRob} can be as well  positive as negative.
   The spectrum of its quantum counterpart is not bounded neither from above, nor from below. 
     
   This system also involves the second integral of motion,
     \be
     \lb{N}
     N \ =\ \frac {P^2}2 + V(x) \,.
      \ee
      The Poisson bracket $\{H,N\}$ vanishes, and
      this means that this system is exactly solvable and its classical trajectories can be explicitly expressed in terms of certain elliptic functions and the integrals thereof. The trajectories do not exhibit collapse: the solutions to the equations of motion exist for all times.
      
      If the classical system is benign, its quantum counterpart is also benign \cite{obzor}. Indeed, the quantum effects bring about quantum fluctuations, which may prevent the system to run into a singularity. This happens e.g. to the 3-dimensional Hamiltonian
       \be
       \lb{H-center}
       H \ =\ \frac {\vec{p}^2}{2m} - \frac \kappa {r^2} \, .
        \ee
        For any nonzero $\kappa$, there are infinitely many classical trajectories (all the trajectories with negative energies)
 that fall on the  center in finite time. One of them is depicted in Fig. \ref{spiral}. But the corresponding quantum problem is benign and the spectrum of the Hamiltonian is well defined if 
        $\kappa$ does not exceed the threshold
          \be
          \kappa^* \ =\ \frac 1{8m} \,.
           \ee
           For still larger values of $\kappa$, quantum fluctuations cannot cope with the highly attractive potential, and the system collapses \cite{padenie}. 
           
           Many ghost-ridden systems are associated with higher-derivative Lagrangians. In fact, {\it any} nondegenerate higher-derivative system involves ghosts \cite{Woodard,eston}. The inverse is not true:  the Lagrangians corresponding to the Hamiltonians \p{HamRob}, \p{H-center} do not include higher derivatives. 
           
           The spectral problem for the quantum counterpart of \p{HamRob} can be exactly solved. There is an infinite number of localized states with zero energy and also the bands of continuous spectrum:
           $E \in [\omega, \infty); (-\infty, -\omega]; [2\omega, \infty); (-\infty, -2\omega]$, etc.
      
      \begin{figure} [ht!]
      \bc
    \includegraphics[width=.5\textwidth]{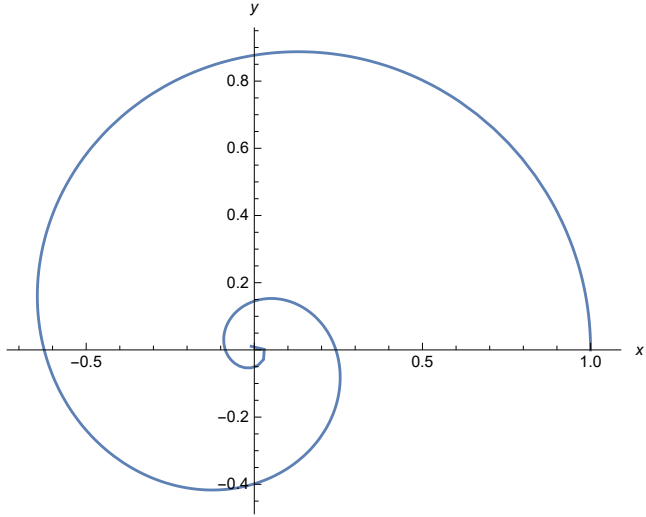}                  
     \ec
    \caption{Falling on the center in the potential \p{H-center} with $m=1$ and $\kappa = .05$. The energy is slightly negative. The particles with positive energies escape to infinity.}        
 \label{spiral}
    \end{figure}  
   
In Ref. \cite{Robert}, also the system representing a perturbation of the system \p{HamRob},
\be
\lb{HamRob-pert}
 H = pP + DV'(x) - \frac \gamma 2 (D^2 + P^2),
 \ee
     was studied. This system is not exactly solvable anymore, but a numerical analysis exhibits a benign behaviour of classical trajectories. They do not collapse.
     
     Later, some other mechanical systems with benign ghosts were  found \cite{KovPav}.
     
     There are also field theory systems with similar properties. We presented an example of such system in \cite{duhi-pole}. Its    Lagrangian reads
     \be
     \lb{duhi-field}
     {\cal L} \ =\ \partial_\mu \phi \partial_\mu D - D V'(\phi),
      \ee
      where $\phi$ and $D$ are now field variables depending on $x$ and $t$; $\mu = 0,1$. This system represents the (1+1)-dimensional generalization of \p{HamRob}. It is not exactly solvable, and the behaviour of its classical trajectories can be studied only numerically. We performed this study and found no trace of collapse: the ghosts are benign!
      
      For sure, when there is no exact analytic solution and the equations of motion can be solved only numerically, one is never sure: for {\it some} initial conditions the trajectory may still collapse. But 
      for the exactly solvable system \p{HamRob}, we {\it are} sure about its benign nature.
     
  \section{Toda chain}
  
  Our main observation is that, besides \p{HamRob}, there exist many other exactly solvable systems with benign ghosts.
  
  Consider as an example the closed Toda  chain with three particles. The Hamiltonian is 
   \be
   \lb{HamToda}
   H \ =\ \frac 12 (p_1^2 + p_2^2 + p_3^2) + V_{12} + V_{23} + V_{31} \,,
      \ee
      where $V_{12} = e^{q_1 - q_2}$, etc. Besides the energy, the system involves an obvious integral of motion $P = p_1 + p_2 + p_3$, as well as the less obvious cubic invariant
       \be 
       \lb{I3}
       I \ =\ \frac 13(p_1^3 + p_2^3 + p_3^3) + p_1 (V_{12} + V_{31} ) + 
       p_2 (V_{12} + V_{23} ) + p_3 (V_{23} + V_{31} ) \,.
        \ee
        The Poisson brackets $\{H, I\}$,   $\{H, P\}$ and $\{I, P\}$ vanish. For the system with three degrees of freedom, we have three integrals of motion that are in involution.  The system is exaclty solvable.
        
        In contrast to the open Toda chain, where the potential represents the sum of only two terms $V_{12} + V_{23}$, the potential in \p{HamToda} keeps the three particles together. If we impose the requirement $q_1 + q_2 + q_3 = 0$ (the center of mass does not move), the motion is finite, with the classical trajectories representing complicated nonlinear oscillations. In the quantum problem \cite{Gutz}, we impose the constraint $\hat P \Psi(q_1, q_2, q_3) = 0$, in which case the spectrum of $\hat H$ is discrete. The eigenfunctions of \p{HamToda}   are simultaneously the eigenfunctions of $\hat I$. \footnote{The operator $\hat I$ is defined by Eq. \p{I3} where all the momenta $\hat p_j$ stay consistently on the left or consistently on the right: the two orderings give the same result.}
        
        Now note that $\hat I$ is cubic in momenta and its eigenvalues can be both positive and negative. The following simple property holds:
        
        \begin{prop}
        For each state with energy $E$ and positive eigenvalue of $\hat I$, there exists a state with the same energy and the negative  eigenvalue of $\hat I$.
         \end{prop}
\begin{proof}
The operator $\hat H$ is invariant under the following discrete symmetry:\footnote{It is also invariant under cyclic permutations of $q_{1,2,3}$, but this is irrelevant for our purposes.}
\be
\lb{sym}
\hat P: \quad q_1 \leftrightarrow -q_2, \ q_3 \to -q_3 \,.
 \ee
 The Hamiltonian commutes with $\hat P$, so that, if $\Psi$ is an eigenfunction of $\hat H$, 
 $\hat P \Psi$ is also an eigenfunction of $\hat H$ with the same energy. But $\hat I$ {\it anticommutes} with   $\hat P$. It follows that, if $\hat I \Psi = \lambda \Psi$, then 
 $$\hat I \hat P \Psi \ =\ - \hat P \hat I \Psi \ =\  -\lambda \hat P \Psi\,.$$
  \end{proof}
 
 In other words, the spectrum of $\hat I$ is symmetric under reflection $\lambda \to -\lambda$.
 To prove that this spectrum extends indefinitely down to $-\infty$ (and extends indefinitely up to $\infty$), it is sufficient to consider the {\it classical} trajectories with large values of $|I|$. The highly excited eigenstates of $\hat H$ and $\hat I$ are related to these trajectories by WKB correspondence. But it is obvious that both the energy and $|I|$ can be arbitrary large. Simply take the initial conditions with $q_j(0) = 0$ and $p_2(0) = p_3(0) = - p_3(0)/2 = M $.  
 
 There is only one step to go. We {\it call} $\hat I$ rather that $\hat H$ the Hamiltonian. The new Hamiltonian does not have a ground state and is hence ghost-ridden. On the other hand, the spectral problem for $\hat I$ is quite well defined. There is no collapse and no loss of unitarity.
 
 A nice feature of the Hamiltonian $\hat I$, compared to the quantum counterpart of \p{HamRob} is its discrete spectrum.
 
 \subsection{Classical dynamics}
 For the quantum Toda problem, it does not matter which operator, $\hat H$ or $\hat I$, is called the Hamiltonian: they have the same spectrum. But the classical dynamics associates with $H$ and $I$ are different. For some reason, the scholars have not been so much interested in the latter. I only was able to find in the literature the Hamilton equations of motion associated with the cubic invariant for the open Toda chain \cite{Perelomov}. In our case, the Hamilton equations of motion for $I$ read 
  \be
  \lb{eqmot-I}
  \dot{q}_1  = p_1^2 + V_{12} + V_{31}, \qquad \dot{p}_1 = (p_1 + p_3) V_{31} - (p_1 + p_2) V_{12}\,, \nn
  \dot{q}_2  = p_2^2 + V_{23} + V_{12}, \qquad \dot{p}_2 = (p_2 + p_1) V_{12} - (p_2 + p_3) V_{23}\,, \nn
  \dot{q}_3  = p_3^2 + V_{31} + V_{23}, \qquad \dot{p}_3 = (p_3 + p_2) V_{23} - (p_3 + p_1) V_{31}\,,
   \ee
   to be compared with the standard equations of motion
   \be
   \lb{eqmot}
   \ddot{q}_1 = \dot{p}_1 = V_{31} - V_{12}, \qquad  \ddot{q}_2 = \dot{p}_2 = V_{12} - V_{23}, \qquad  \ddot{q}_3 = \dot{p}_3 = V_{23} - V_{31}\,. 
     \ee
 To feel the difference, compare the classical trajectories associated with $H$ and $I$ (see  Fig. 2).
 
 \begin{figure}[ht!]
 \lb{clas-HI1}
     \begin{center}

        \subfigure[$H$]{
           \lb{clas-H}  \includegraphics[width=0.4\textwidth]{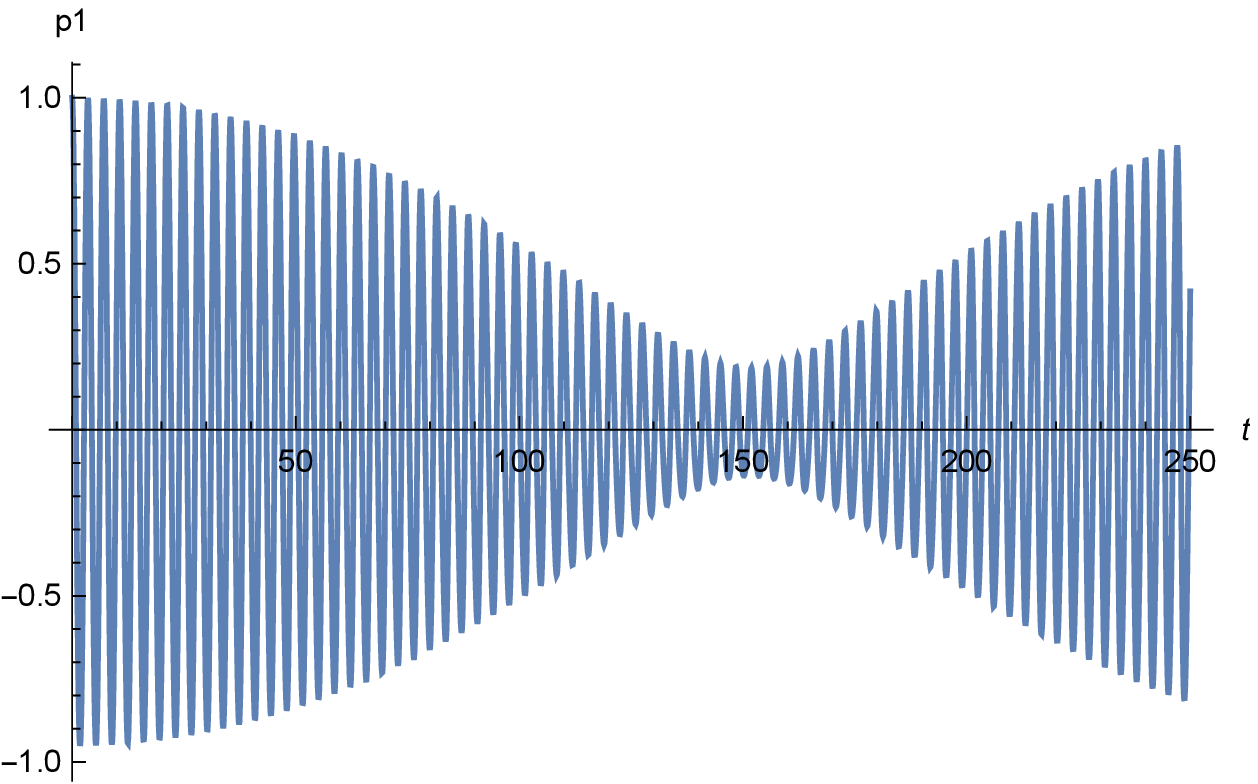}
        }
        \subfigure[$I$]{
         \lb{clas-I}
           \includegraphics[width=0.4\textwidth]{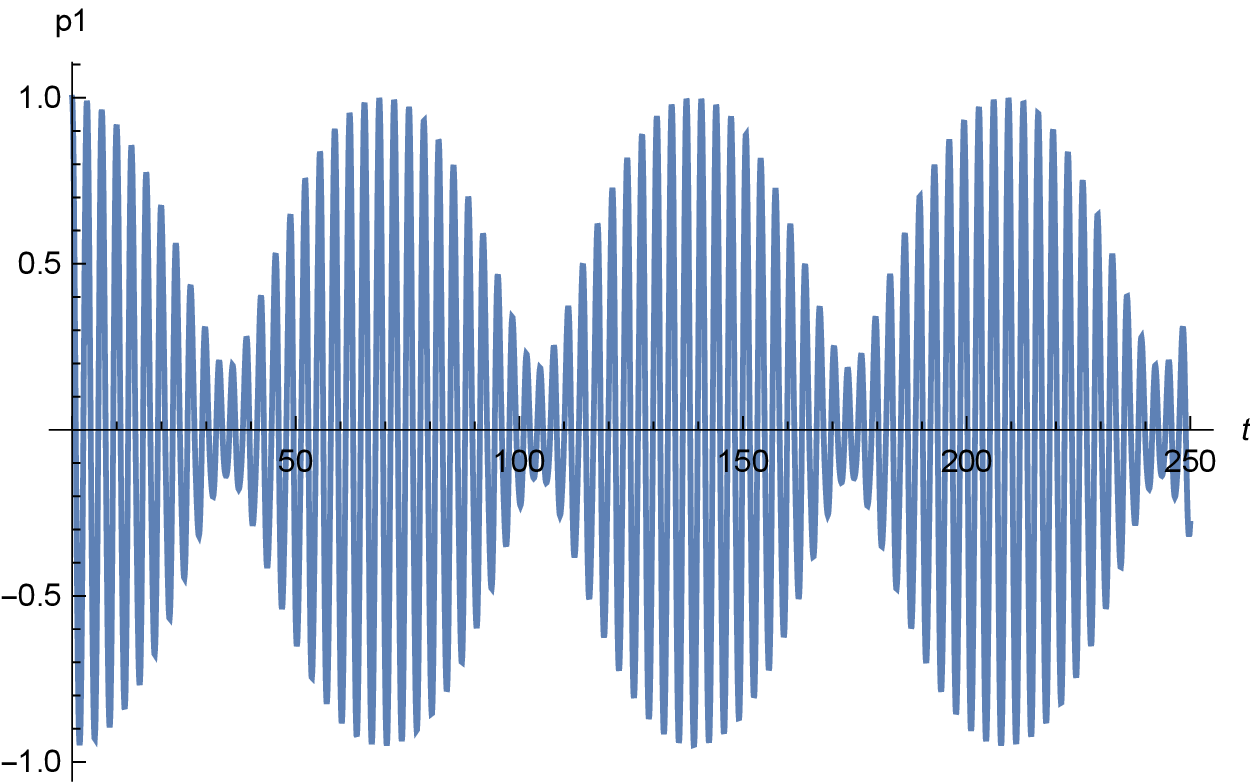}
        }
    \end{center}
    \caption{ The dependence $p_1(t)$ for the equations of motion based on $H$ and on $I$. The inital conditions  are  $q_1(0) = q_2(0) = q_3(0) =0, p_1(0) = 1, p_2(0) = p_3(0) = -.5$.}
    \end{figure}
    
    The initial conditions for both trajectories are identical and the conserved values of $H$ and $I$ are the same. The patterns of the two trajectories are similar---oscillations with beats. But for the equations of motion based on $I$, the period of beats is smaller. \footnote{There is no wonder that the classical trajectories associated with $H$ and $I$ are different. The phase space has four essential coordinates: $p_{1,2}, q_{1,2}$. The fixed values of two integrals of motion define a two-parametric surface in this space. Both trajectories lie on this surface, but they need not coincide. }
    
    The solutions to \p{eqmot} for $q_j(t)$   also oscillate, while the solutions  to \p{eqmot-I} grow quasilinearly. This difference is due to the fact that, for the system \p{eqmot-I}, the velocities do not coincide with the momenta.
    
     \begin{figure}[ht!]
 \lb{traj-I-tilde}
       \includegraphics[width=0.5\textwidth]{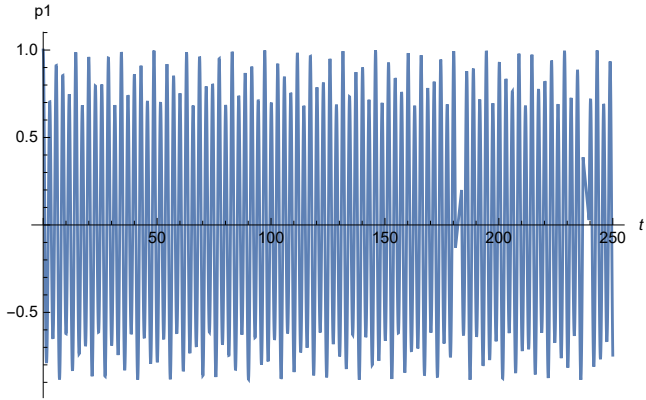}
    \caption{A typical trajectory for the Hamiltonian \p{tilde-I}. }
    \end{figure}
    
    Let us now perturb our system by adding to the new Hamiltonian $I$ an oscillatory potential:
     \be
     \lb{tilde-I}
     I \ \to \ \tilde{I} \ =\ I + \frac \alpha 2 [(q_1 - q_2)^2 + (q_1 - q_3)^2 + (q_2 - q_3)^2]\,.
      \ee
      The Poisson bracket $\{\tilde{I}, p_1 + p_2 + p_3\}$ still vanish, but the bracket $\{ \tilde{I}, H\}$ does not. Thus, $H$ is not an integral of motion anymore and the system is not integrable.
      
      The trajectories can be found numerically. The trajectory for the Hamiltonian \p{tilde-I} with $\alpha=1$ and the same initial conditions as in Fig. 2 is shown in Fig. 3. We observe quasiperiodic oscillations (the regularity of trajectories seen in Fig. 2 is now lost, which confirms the fact that the system is no longer exactly solvable). We have not observed a collapse though, with numerical calculations, one can never be sure.
      
    \subsection{Path integral}
    
    For a ghost-ridden system, the {\it Euclidean} path integral is not defined \cite{no-Euclid,obzor}. Indeed, if the spectrum does not have a bottom, the sum  for the Euclidean evolution operator,
    \be
    \lb{ev-Eucl}
   {\cal K}  (q^{\rm f}, q^{\rm in}; \beta) \ =\ \sum_n \Psi^*_n(q^{\rm f}) 
   \Psi_n(q^{\rm in}) e^{-\beta I_n}\,,
    \ee
    badly diverges. However, if the time is real, the presence of negative energies does not lead to trouble: Minkowski path integral is well defined if the classical trajectories are well defined. By the same token as for the ordinary quantum systems studied in \cite{Feynman}, the path integral for the evolution operator for an exactly solvable higher-derivative system is saturated by the trajectories close to the classical trajectory, which gives
     \be
     \lb{ev-Mink}
      {\cal K}  (q^{\rm f}, q^{\rm in}; T) \ =\  F(T)\exp\{ i S_{\rm cl}\}  \,,
       \ee
       where $S_{\rm cl}$ is the classical action on the classical trajectory bringing the system from the state $\{q^{\rm in}\}$ to the state $\{q^{\rm f}\}$ in time $T$. The only requirement for the formula \p{ev-Mink} to hold is the existence of the classical trajectories. If they exist, it does not matter whether the quantum spectrum is bounded from below or not. Incidentally, one of the examples treated in the book \cite{Feynman}   [see  Eq. (3.62) there] is the motion in the linearly growing potential. The ground state in this case is absent, but the evolution operator is still well defined.
     
       The factor $F(T)$, which comes from the integral over the deviations from the classical trajectory, is easy to determine when $S_{\rm cl}$ is large and
       the integral can be evaluated in the  Gaussian WKB  approximation.   We suspect that, for an exactly solvable model, this approximation gives the exact answer, as it does for the harmonic oscillator.
        But  it is a conjecture, of course.
       
       \section{Field theories}
       
       The higher-derivative field theory studied in \cite{duhi-pole} was not exactly solvable and the trajectories, found numerically, exhibit there a stochastic behaviour. But there are many exactly solvable $(1+1)$ systems. They are characterized by an infinite number of integrals of motion. {\it Each} of them can be chosen as a Hamiltonian, and we obtain thereby a set of higher-derivative exactly solvable $(1+1)$-dimensional field systems. Take as an example the Sine-Gordon model. Its equation of motion reads 
        \be
        \lb{egmot-SGordon}
        \phi_{tt} - \phi_{xx} \ =\ - \sin \phi \,.
          \ee
          The conventional energy of the system,
          \be
          \lb{E-SGordon}
          E \ =\ \int_{-\infty}^\infty   dx \left[ \frac 12 (\phi_t^2 + \phi_x^2) +(1 - \cos \phi) \right]
           \ee
           is conserved, and this gives the standard Hamiltonian.
           
           But there is an infinity of other integrals of motion. The next  in complexity after \p{E-SGordon} is \cite{Lamb}
           \be
           \lb{E4} 
   E_4 \ =\ \int dx \left[ \frac 14 (\phi_t - \phi_x)^4 -  (\phi_{tt} - 2 \phi_{tx}
   + \phi_{xx})^2 -    (\phi_t - \phi_x)^2  \,  \cos \phi \right]\,.
    \ee
 This expression involves higher derivatives, is not bounded from below, and the  Hamiltonian $H_4$ is ghost-ridden. To the best of my knowledge, the corresponding dynamical field equations were never studied (not speaking of quantum dynamics!), but such a study would be very interesting. 
 
 As the last example, consider the KdV system. 
 The equation of motion is
  \be
  \lb{eqmot-KdV}
  u_t  + 6 u u_x +  u_{xxx} \ =\ 0\,.
   \ee
   It follows from the Lagrangian
   \be
   \lb{LKdV}
    L \ =\ \frac 12 \psi_{xx}^2 - \psi_x^3 - \frac 12 \psi_t \psi_x\,,
     \ee
     where $\psi_x = u$. The conserved energy of the field reads
     \be
     \lb{EKdV}
     E \ =\ \int dx \left[u^3 - \frac 12 u_x^2\right].
      \ee
      The corresponding local conservation law is
      \be
      \lb{local}
      \frac \partial {\partial t} \left[u^3 - \frac 12 u_x^2\right] \ =\  \frac \partial {\partial x} \left[u_x (6 u u_x + u_{xxx}) - 
      \frac 12 u_{xx}^2 - 3 u^2 u_{xx} - \frac 92 u^4 \right] \,.
       \ee
       There is also an infinite set of higher integrals of motion:
       \be
   \lb{E4-KdV}
   E_4 \ =\ \int  dx \, [5 u^4 + 5 u^2 u_{xx} + u_{xx}^2 ],
    \ee
    etc.  The integrands in \p{E4-KdV}, etc. involve  
    higher-derivatives over $x$, but not over $t$, which does not mean ghosts in the ordinary sense.
    
    But nobody prevents us to {\it interchange} the coordinates $x$ and $t$ and consider the Lagrangian
    \be
   \lb{LKdV-inverse}
    L \ =\ \frac 12 \psi_{tt}^2 - \psi_t^3 - \frac 12 \psi_t \psi_x\,,
     \ee
     giving the equations of motion $u_x  + 6 u u_t +  u_{ttt} = 0$.
     This   Lagrangian includes  higher  temporal derivatives, and the general arguments of Refs. \cite{Woodard,eston} tell us that the corresponding energy functional
      \be
      \lb{En-inverse}
      \int  dx\, \left[ u_t (6 u u_t + u_{ttt}) - 
      \frac 12 u_{tt}^2 - 3 u^2 u_{tt} - \frac 92 u^4  \right]
      \ee
      [cf. the R.H.S. of Eq.\p{local}] is not positive definite, and the spectrum of the corresponding quantum problem is not bounded from below. Such a system involves ghosts, but exact solvability ensures that these ghosts are benign. Indeed, a general solution to the standard KdV equation exists for all $x$ and this suggests that a general solution to the equations of motion following from \p{LKdV-inverse} exists for all $t$. \footnote{We say ``suggests" because we cannot give a rigourous mathematical proof of this assertion.}

 \section*{Acknowledgements}
 I am indebted to Sergei Fedoruk, Joshua Feinberg, Askold Perelomov and Stepan Sidorov for  illuminating discussions.

\end{document}